\documentclass[journal]{IEEEtran}

\usepackage{amsmath,amssymb}
\usepackage{graphicx,graphics,color,psfrag}
\usepackage{cite,balance}
\usepackage{algorithm}
\usepackage{accents}
\usepackage{bm}
\usepackage{url}
\usepackage{algorithmic}
\usepackage[english]{babel}
\usepackage{multirow}
\usepackage{enumerate}
\usepackage{cases}
\usepackage{stfloats}
\usepackage{dsfont}
\usepackage{color,soul}
\usepackage{amsfonts}
\usepackage{tcolorbox}
\usepackage{amsmath}
\usepackage{float}
\usepackage{mathtools}
\usepackage{verbatim} 
\usepackage{bm,color,soul}
\usepackage{epsfig}
\usepackage{subfigure}
\usepackage{psfrag}
\usepackage{latexsym}
\usepackage[footnotesize]{caption}
\usepackage{color}
\usepackage{booktabs}
\usepackage[caption=false,font=normalsize,labelfont=sf,textfont=sf]{subfig}

\IEEEoverridecommandlockouts

\usepackage{algorithm}
\usepackage{algorithmic}
\usepackage{cuted}
\usepackage{lipsum}
\usepackage{stfloats}
\usepackage{cite,graphicx,amsmath,amssymb}
\usepackage{fancyhdr}
\usepackage{hhline}
\usepackage{array,color}
\usepackage{amsmath}
\usepackage{stfloats}
\usepackage[flushleft]{threeparttable}
\usepackage{makecell} 
\newtheorem{prop}{Proposition}

\newtheorem{proof}{proof}

\begin{document}

%
\title{Networked ISAC for Low-Altitude Economy: Transmit Beamforming and UAV Trajectory Design}

\author{
\IEEEauthorblockN{Gaoyuan~Cheng, Xianxin~Song, Zhonghao~Lyu, and  Jie~Xu%
\thanks{ J. Xu is the corresponding author. }
}

\IEEEauthorblockA{School of Science and Engineering (SSE) and Future Network of Intelligence Institute (FNii), \\ The Chinese University of Hong Kong, Shenzhen, Guangdong, China} \\
Email: \{gaoyuancheng, xianxinsong, zhonghaolyu\}@link.cuhk.edu.cn, xujie@cuhk.edu.cn
}


\maketitle


\begin{abstract}
This paper studies the exploitation of networked integrated sensing and communications (ISAC) to support low-altitude economy (LAE), in which a set of networked ground base stations (GBSs) transmit wireless signals to cooperatively communicate with multiple authorized unmanned aerial vehicles (UAVs) and concurrently use the echo signals to detect the invasion of unauthorized objects in interested airspace. Under this setup, we jointly design the cooperative transmit beamforming at multiple GBSs together with the trajectory control of authorized UAVs and their GBS associations, for enhancing the authorized UAVs' communication performance while ensuring the sensing requirements for airspace monitoring. In particular, our objective is to maximize the average sum rate of authorized UAVs over a particular flight period, subject to the minimum illumination power constraints for sensing over the interested airspace, the maximum transmit power constraints at individual GBSs, and the flight constraints at UAVs. This problem is non-convex and challenging to solve, due to the involvement of integer variables and the coupling of optimization variables. To solve this non-convex problem, we propose an efficient algorithm by using the techniques of alternating optimization (AO), successive convex approximation (SCA), and semi-definite relaxation (SDR). Numerical results show that the obtained transmit beamforming and UAV trajectory designs in the proposed algorithm efficiently balance the tradeoff between the sensing and communication performances, thus significantly outperforming various benchmarks.

\end{abstract}



\section{Introduction}
Low-altitude economy (LAE) corresponds to a comprehensive economic form consisting of various low-altitude flight activities of unmanned and manned aircraft such as unmanned aerial vehicles (UAVs) and electric vertical take-off and landing (eVTOL). LAE is expected to enable a series of low-altitude applications in transportation, environmental monitoring, agriculture, and entertainment to create great economic and social value, which has attracted explosively increasing research attentions recently. How to ensure seamless wireless communication connections with authorized aircraft and how to monitor the airspace to prevent the invasion of unauthorized objects are essential for the success of LAE \cite{Whitepaper4}.

As one of the key technologies for six-generation (6G) wireless networks, integrated sensing and communications (ISAC) has emerged as an efficient solution to support LAE \cite{Whitepaper4, Mu2023magazine}, in which
ground base stations (GBSs) can use the transmitted wireless signals to communicate with authorized aircraft as network-connected aerial users and also receive echo signals for sensing low-altitude airspace and monitoring the invasion of unauthorized objects. More specifically, unlike conventional mono-static sensing via isolated transceivers, GBSs in 6G networks can enable networked ISAC to provide ubiquitous sensing and communication over the three-dimensional (3D) space by exploiting their cooperation in both coordinated multi-point (CoMP) communication and distributed multiple-input multiple-output (MIMO) radar sensing \cite{zhang2021perceptive, Fan2023Sensing}. Networked ISAC offers several advantages for LAE. On the one hand, networked ISAC can implement the joint signal processing for cooperative transmission among GBSs, thus properly mitigating or even utilizing the multi-cell air-ground interference. On the other hand, distributed GBSs can exploit the spatial diversity of the target radar cross section (RCS) from different sensing directions to enhance the sensing and detection performance \cite{Fei2023magazine}. Therefore, it is becoming increasingly important to exploit networked ISAC to support real-time communication and tracking of authorized aircraft , and sense the airspace to monitor the invasion of unauthorized objects.

In the literature, there have been some prior works studying the interplay between ISAC and UAVs  \cite{Meng2023magazine, Meng2023Throughput, Lyu2022Joint,wang2020constrained, Wu2023uav}. For instance, the authors in \cite{Lyu2022Joint} employed a low-altitude UAV as an aerial ISAC platform to communicate with ground users and sense potential targets at interested areas, in which a novel transmit beamforming and UAV trajectory design algorithm was proposed to balance the trade-off between communication and sensing performances. In \cite{wang2020constrained}, the authors proposed an ISAC system with multiple UAVs acting as aerial base stations (BSs) and communicating with ground users while sensing a target cooperatively. The UAVs' transmit power and locations as well as user association were jointly optimized for maximizing the network utility. The authors in \cite{Wu2023uav} investigated a UAV-enabled ISAC system with a set of UAVs sending communication signals to one mobile ground user, in which the UAV trajectory and user-UAV association were jointly designed. However, the above prior works only considered using UAV-enabled ISAC platforms to support ISAC for on-ground users and targets. To the best of our knowledge, the research on using on-ground networked ISAC systems to support LAE with enhanced ISAC performances has not been well investigated yet. In particular, how to jointly design the transmit beamforming of GBSs and the trajectory design of authorized UAVs to optimize the communication performance while ensuring the requirements on monitoring targeted airspace is an unaddressed but challenging problem.


This paper studies a typical scenario of networked ISAC to support LAE, in which a set of networked GBSs send combined information and sensing signals to  cooperatively communicate with multiple authorized UAVs and simultaneously sense an interested airspace to monitor the invasion of unauthorized objects. We aim to jointly optimize the cooperative transmit information and sensing beamforming of GBSs, the trajectory design of authorized UAVs, and their association with GBSs to maximize the average sum rate of authorized UAVs over a particular flight period, subject to the transmit power constraints at GBSs, the practical flight constraints at UAVs, the GBS-UAV association constraints, and the minimum illumination power constraints for sensing over interested airspace. However, the above problem is difficult to solve due to the involvement of integer variables and the coupling of optimization variables. To address this issue, we propose an efficient algorithm to find a high-quality solution by using the techniques of alternating optimization (AO),  semidefinite relaxation (SDR), and successive convex approximation (SCA). Finally, numerical results show that the proposed transmit beamforming and UAV trajectory designs efficiently enhance the ISAC performance and significantly outperform various benchmark schemes.



\begin{figure}[ht]
\centering
    \includegraphics[width=5.5cm]{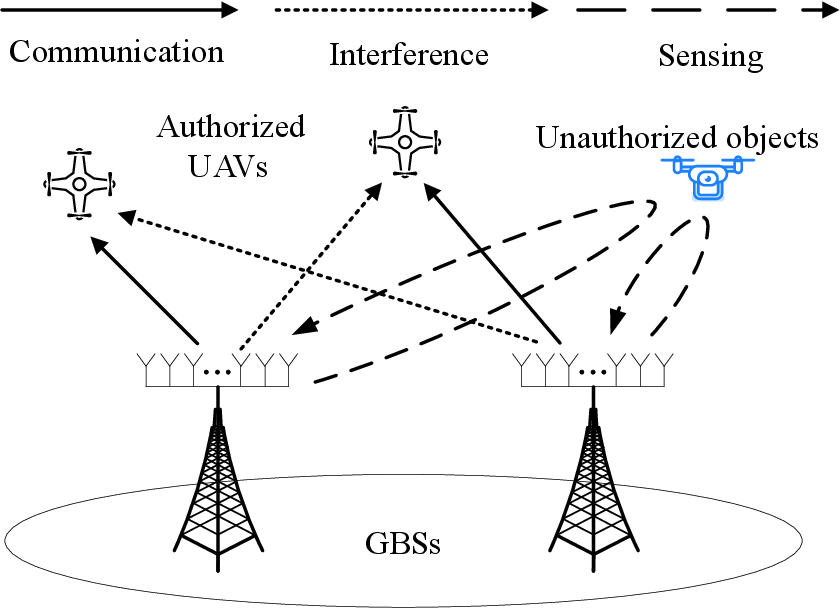}
\caption{Illustration of networked ISAC for communicating with authorized UAVs and monitor the invasion of unauthorized objects.}
\label{fig:system}
\end{figure}

\section{System Model and Problem Formulation}
We consider a networked ISAC system as illustrated in Fig. \ref{fig:system}, which consists of $M$ GBS each  with $N_a$ antennas and $K$ authorized UAVs each with a single antenna. The sets of GBSs and authorized UAVs are denoted as ${\cal M} = \{1,\ldots,M \}$ and ${\cal K} = \{1,\ldots,K \}$, respectively. In this system, the GBSs perform downlink communication with their  associated authorized UAVs and concurrently sense the targeted airspace for monitoring the unauthorized objects like UAVs.

We focus on the ISAC transmission over a particular time period with ${\cal T}=[0, T]$ with duration $T$, which is divided into $N$ time slots each with duration ${\Delta _t} = T/N$. Let ${\cal N}=\{1, \ldots, N\}$ denote the set of slots. Here, $N$ is chosen to be sufficiently large and accordingly ${\Delta _t}$ is sufficiently small, such that the UAVs' locations are assumed to be unchanged over each slot to facilitate the system design. Without loss of generality, we consider a 3D Cartesian coordinate system. Let ${\bf u}_m=(x_m,y_m)$ denote the horizontal coordinate of GBS $m \in {\cal M}$, and ${\bf q}_k[n]= ({\hat x}_k[n], {\hat y}_k[n])$ denote the time-varying horizontal coordinate of UAV $k \in {\cal K}$ at slot $n \in {\cal N}$. All GBSs are located at the zero altitude, and each UAV $k \in {\cal K}$ is at a fixed altitude of $H_k>0$, where $H_k$'s can be different among different UAVs due to their pre-assigned flight region.

At each slot $n$, each UAV is associated with one single GBS. We use a binary variable ${\alpha}_{m,k}[n] \in \{0,1\}$ to indicate the association relationship between GBS $m$ and UAV $k$ at slot $n$. Here, we have ${\alpha}_{m,k}[n]=1$ if UAV $k$ is associated with GBS $m$ at slot $n$ and ${\alpha}_{m,k}[n]=0$ otherwise. As such, we have $\sum\nolimits_{l \in {\cal M}} {{{\alpha}_{l,k}}[n]}  = 1, \forall k \in {\cal K}, n \in {\cal N}$. The GBSs send combined information and dedicated sensing signals to achieve full degrees-of-freedom (DoFs) for sensing. Let $s_{m,k}[n]$ denote the transmit information signal sent by GBS $m$ to UAV $k$ at time slot $n$, and ${\bf w}_{m,k}[n]$ denote the corresponding transmit beamforming vector. Here, we assume $s_{m,k}[n]$'s are independent and identically distributed (IID) random variables each with zero mean and unit variance. Let ${\bf s}_{m}[n] \in {\mathbb C}^{N_a \times 1}$ denote the dedicated sensing signal vector sent by GBS $m$ at time slot $n$ with covariance matrix ${{\bf{R}}_m}[n]= {\mathbb E}({{\bf{s}}_m}[n]{\bf{s}}_m^H[n]) \succeq {\bf 0}$, where ${\mathbb E}$ denotes the statistic expectation and the superscript $H$ denotes the conjugate transpose. As a result, the transmitted signal by GBS $m$ at time slot $n$ is
\begin{align}
{{\bf{x}}_m}[n] = \sum\nolimits_{i \in {\cal K}} {{\bf{w}}_{m,i}}[n]{s_{m,i}}[n]  + {{\bf{s}}_m}[n], n \! \in \!{\cal N},
\end{align}
and the statistical covariance matrix of ${\bf x}_m[n]$ is
\begin{align}
{{\bf{X}}_m}[n] =\sum\nolimits_{i \in {\cal K}}  {\bf{w}}_{m,i}[n]{{\bf{w}}_{m,i}^H}[n] + {{\bf{R}}_m[n]}.
\end{align}

\subsection{Communication Model}

First, we consider the communication from GBSs to UAVs. Let ${{\bf{h}}_{m,k}}[n] = \sqrt {{\beta _{m,k}}[n]} {{\bf{g}}_{m,k}}[n]$ denote the channel vector between GBS $m$ and UAV $k$, where ${\beta _{m,k}}[n] = \kappa ({ {\left\| {{{\bf{q}}_k}[n] - {{\bf{u}}_m}} \right\|}^2+ {H_k^2}}) ^{ - 1}$, and $\kappa$ denotes the path loss at the reference distance of one meter. In this model, we assume the air-ground links from GBSs to UAVs are dominated by light-of-sight (LoS) channels. By assuming that uniform linear arrays (ULAs) are deployed at GBSs, we have
\begin{align}
{{{\bf{g}}}_{m,k}}[n] &= [1,{e^{j2\pi \frac{d}{\lambda }\cos \theta ({{\bf{q}}_k}[n],{{\bf{u}}_m})}}, \ldots ,\nonumber \\
&{e^{j2\pi \frac{d}{\lambda }({N_a} - 1)\cos \theta ({{\bf{q}}_k}[n],{{\bf{u}}_m})}}]^T, \label{steer:1}
\end{align}
where the superscript $T$ denotes the transpose,  $j=\sqrt{-1}$, $\lambda$ denotes the carrier wavelength, $d$ denotes the antenna spacing, and  ${\theta ({{\bf{q}}_k}[n],{{\bf{u}}_m})}$ denotes the angle of departure (AoD) between GBS $m$ and UAV $k$ that is given by
\begin{align}
\theta ({{\bf{q}}_k}[n],{{\bf{u}}_m}) = \arccos \frac{H_k}{({{\left\| {{{\bf{q}}_k}[n] - {{\bf{u}}_m}} \right\|}^2 + H_k^2})^{\frac{1}{2}} }.
\end{align}
As a result, if UAV $k$ is associated with GBS $m$ at slot $n$ or equivalently ${\alpha _{m,k}}[n] = 1$, then the corresponding received signal-to-interference-plus-noise (SINR) at UAV $k$ is is denoted by $\gamma_{m,k}$ in \eqref{SINR} at the top of the next page,
\begin{figure*}
\begin{align} \label{SINR}
\gamma _{m,k}[n] = \frac{ {\left\| {{\bf{h}}_{m,k}^H[n]{{\bf{w}}_{m,k}}[n]} \right\|}^2 }{{\sum\nolimits_{(l,i) \ne (m,k)} {{ || {{\bf{h}}_{m,k}^H[n]{{\bf{w}}_{l,i}}[n]} || }^2 }  + \sum\nolimits_{l \in {\cal M}} {{\bf{h}}_{m,k}^H[n]{{\bf{R}}_l}[n]{{\bf{h}}_{m,k}[n]}}  + \sigma^2}}.
\end{align}
\end{figure*}
where $\sigma ^2$ denotes the power of the additive withe Gaussian noise (AWGN) at the UAV receiver. It is observed from \eqref{SINR} that each UAV suffers from interference not only from other communication signals but also the dedicated sensing signals. Let ${r_{m,k}}[n] = {\log _2}(1 + {\gamma _{m,k}}[n])$ denote the achievable rate if UAV $k$ is associated to GBS $m$ in slot $n$. Then the achievable sum rate of the $K$ authorized UAVs in slot $n$ is
\begin{align}
r[n]= \sum\nolimits_{m \in {\cal M}}\sum\nolimits_{k \in {\cal K}} {\alpha _{m,k}}[n]{r_{m,k}}[n] .
\end{align}

\subsection{Sensing Model}

Next, we consider radar sensing towards the targeted airspace. For facilitating the ISAC design, we sample $Q$ locations as representative sensing points at the corresponding airspace, each of which has an altitude of $H$ and a horizontal location of ${\bf v}_q$, $q\in {\cal Q} \buildrel \Delta \over = \{ 1, \ldots ,Q\}$. As such, the distance between GBS $m$ and sensing location $q$ is ${d_{l,q}} = \sqrt{  {{{\left\| {{{\bf{u}}_l} - {{\bf{v}}_q}} \right\|}^2} + {H^2}}} $. Let $\theta_{m,q} = \arccos \frac{H}{{d_{l,q}}}$ denote the corresponding angle between GBS $m$ and sensing location $q$. We have the corresponding steering vector towards the $q$-th sensing location as
\begin{align}
{\bf{a}}({\theta _{m,q}}) \! =\! [1,{e^{j2\pi \frac{d}{\lambda }\cos ({\theta _{m,q}})}}, \ldots ,{e^{j2\pi \frac{d}{\lambda }({N_a} - 1)\cos ({\theta _{m,q}})}}]^T.
\end{align}
We use the illumination (or received) signal power at the interested sensing locations as the sensing performance metric. For target sensing location $q$, the illumination power or the received power from the $M$ GBSs at slot $n$ is
\begin{align}
{\zeta _{q}}[n] &= \sum\nolimits_{l \in {\cal M}}   {{\bf{a}}^H}(\theta _{l,q})(\sum\nolimits_{i \in {\cal K}} {{\bf{w}}_{l,i}^H[n]{{\bf{w}}_{l,i}}[n]} \nonumber \\
& + {{\bf{R}}_l}[n]){{\bf{a}}^H}(\theta _{l,q})/{d_{l,q}}^2.
\end{align}

\subsection{Problem Formulation}
We aim to maximize the sum rate at all the $K$ authorized UAVs over the $N$ time slots, by jointly optimizing the cooperative transmit beamforming $\{ {\bf{w}}_{l,i}[n]\}$ and $\{ {\bf{R}}_{l}[n]\}$ at GBSs, the trajectory design $\{{\bf{q}}_k[n]\}$ of authorized UAVs, and the UAV-GBS association $\{\alpha _{m,k}[n]\}$, subject to the minimum illumination power constraints towards the targeted airspace, and the maximum transmit power constraints at GBSs. In particular, we assume the initial and final locations of UAVs are fixed to be ${\bf{q}}_k[1] = {\bf{q}}_k^{{\rm{I}}}$ and ${\bf{q}}_k[N] = {\bf{q}}_k^{{\rm{F}}}, \forall k \in {\cal K}$, respectively. Also, the UAVs' flights are subject to their individual maximum speed constraints and the collision avoidance constraints, i.e.,
\begin{align}
\left\| {{{\bf{q}}_k}[n + 1] - {{\bf{q}}_k}[n]} \right\| \le {V_{\max }}{\Delta _t}, \forall k \in {\cal K}, n \in {\cal N},
\end{align}
\begin{align}
&\left\| {{{\bf{q}}_k}[n] - {{\bf{q}}_i}[n]} \right\|^2 + (H_k -H_i)^2 \ge {D_{\min }^2}, \nonumber \\
&\forall k ,i \in {\cal K}, k \ne i,n \in {\cal N},
\end{align}
where $V_{\max}$ denotes the maximum UAV speed, and ${D_{\min }}$ denotes the minimum distance between any two UAVs for collision avoidance. Furthermore, to ensure the sensing requirements, we suppose that the illumination power towards each targeted sensing location $q$ should be no less than a given threshold $\Gamma$. As such, the joint cooperative transmit beamforming and UAV trajectory optimization problem for the LAE-oriented networked ISAC system is formulated as
\begin{subequations}
\begin{align}
\left( {\rm{P1}} \right):&\mathop {\max }\limits_{\{ {{\bf{w}}_{m,k}[n]},{{\bf{R}}_m[n]},{\bf q}_k[n], {\alpha _{m,k}}[n]\} }~ \sum\nolimits_{n \in {\cal N}} {{r}[n]} \label{p1:obj} \\
{\rm{s.t.}}&~  {{\zeta _{q}}[n]}  \ge \Gamma  , \forall q, n, \label{p1:bp} \\
&~\sum\nolimits_{i \in {\cal K}} {\left\| {{\bf{w}}_{m,i}}[n] \right\|}^2 \! +\! {\rm{tr}}\left( {{\bf{R}}_m}[n] \right) \le {P_{\max }}, \forall m, \label{p1:pow} \\
&~{\bf{q}}_k[1] = {\bf{q}}_k^{{\rm{I}}}, \forall k,\label{p1:ini} \\
&~{\bf{q}}_k[N] = {\bf{q}}_k^{{\rm{F}}}, \forall k,\label{p1:end} \\
&~\left\| {{{\bf{q}}_k}[n + 1] - {{\bf{q}}_k}[n]} \right\|^2 \le ({V_{\max }}{\Delta _t})^2, \forall k, n, \label{p1:speed} \\
&~\left\| {{{\bf{q}}_k}[n] - {{\bf{q}}_i}[n]} \right\|^2 + (H_k -H_i)^2 \ge {D_{\min }^2}, \nonumber \\
&~\forall k,i,n, k \ne i, \label{p1:crash} \\
&~{\alpha}_{m,k}[n] \in \left\{ {0,1} \right\},  \forall m, k , n, \label{p1:asso1} \\
&~\sum\nolimits_{l \in {\cal M}} {{{\alpha}_{l,k}}}[n]= 1, \forall k, n. \label{p1:asso2}
\end{align}
\end{subequations}

However, problem (P1) is very difficult to solve in general. This is due to the fact that the sum rate objective function in \eqref{p1:obj} is highly non-concave, the collision avoidance constraints in \eqref{p1:crash}, the binary UAV association constraints in \eqref{p1:asso1} are non-convex, and the optimization variables $\{ {\bf{w}}_{l,i}[n]\}$, $\{ {\bf{R}}_{l}[n]\}$, $\{{\bf{q}}_k[n]\}$, and $\{\alpha _{m,k}[n]\}$ are coupled. Therefore, problem (P1) is a mixed-integer non-convex problem that is challenging to solve.

\section{Proposed Solution to Problem (P1)} \label{Sec:solution}
In this section, we propose an efficient algorithm to solve problem (P1) by using AO. In particular, we alternately optimize the UAV association $\{ \alpha_{m,k}[n]\}$, transmit beamforming $\{{{\bf{w}}_{m,k}}[n]\}$, $\{{{\bf{R}}_{m}}[n]\}$ at GBSs, and UAV trajectory $\{{{\bf{q}}_k}[n]\}$ via using the techniques of SCA and SDR.

To facilitate the derivation, we define ${\bf{H}}_{m,k}[n] = {\bf{h}}_{m,k}[n]{\bf{h}}_{m,k}^H[n]$ and ${{\bf{W}}_{m,k}}[n] = {{\bf{w}}_{m,k}}[n]{\bf{w}}_{m,k}^H[n]$ with ${\bf W}_{m,k}[n] \succeq 0$ and ${\text {rank}}({\bf W}_{m,k}[n]) \le 1$, $\forall m,k,n$ . Accordingly, $ r_{m,k}[n]$ is reformulated in \eqref{sumrate:SDR} at the top of the next page and problem (P1) is equivalently reformulated as
\begin{figure*}
\begin{align} \label{sumrate:SDR}
{r_{m,k}}[n] = {\log _2}\left(1+ {\frac{{{{\rm{tr}}\left( {\bf{H}}_{m,k}[n]{\bf{W}}_{m,k}[n] \right)} }}{{\sum\nolimits_{(l,i) \ne (m,k)} { {{\rm{tr}}\left( {\bf{H}}_{m,k}[n]{\bf{W}}_{l,i}[n] \right)} }  + \sum\nolimits_{l \in {\cal M}} {\rm{tr}}({\bf{H}}_{m,k}[n]{\bf{R}}_l[n] ) + \sigma _c^2}}} \right).
\end{align}
\vspace{-3ex}
\end{figure*}
\begin{subequations}
\begin{align}
&\left( {\rm{P2}} \right): \max\limits_{ \{ \scriptstyle {\bf{W}}_{m,k}[n]\succeq 0,{{\bf{R}}_m}[n]\succeq 0, \hfill\atop
\scriptstyle {\bf q}_k[n], \alpha _{m,k}[n] \} \hfill}  \sum\nolimits_{n \in {\cal N}} {{r}}[n]  \\
{\rm{s.t.}}&\sum\nolimits_{l \in {\cal M}} {{\bf{a}}^H}(\theta _{l,q})(\sum\nolimits_{i \in {\cal K}} {{{\bf{W}}_{l,i}}[n]}+ {{\bf{R}}_l}[n])\nonumber \\
&\cdot {{\bf{a}}^H}(\theta _{l,q})/{d_{l,q}}^2   \ge  \Gamma  ,\forall q ,n , \label{P2:con:sen} \\
&\sum\nolimits_{i \in {\cal K}} {\rm{tr}}\left( {{\bf{W}}_{m,i}}[n] \right) \! +\! {\rm{tr}}\left( {{\bf{R}}_m}[n] \right) \le {P_{\max }},\forall m, n, \label{P2:con:pow} \\
&{\rm{rank}}\left( {\bf{W}}_{m,k}[n] \right) \le 1, \forall m \in {\cal {M}}, k \in {\cal K} \label{P2:con:rank1} , \\
&\eqref{p1:ini}, ~\eqref{p1:end},~\eqref{p1:speed}, ~\eqref{p1:crash}, ~\eqref{p1:asso1},~{\text{and}}~\eqref{p1:asso2}. \nonumber
\end{align}
\end{subequations}

In the following, we present the optimization of $\{\alpha _{m,k}[n]\}$, $\{{\bf{W}}_{m,k}[n], {{\bf{R}}_m}[n]\}$, and $\{{\bf q}_k[n]\}$, respectively, by assuming the others are given.

\subsection{UAV-GBS Association Optimization}
In this subsection, we optimize the UAV-GBS association $\{{\alpha _{m,k}}[n]\}$ under given trajectory $\{{\bf q}_k[n]\}$ and transmit beamforming $\{{\bf{W}}_{m,k}[n]\}$ and $\{{\bf{R}}_m[n]\}$. As such, the UAV-GBS association optimization problem is
\begin{subequations}
\begin{align}
\left( {\rm{P3}} \right): &\mathop {\max }\limits_{\{ {\alpha _{m,k}}[n]\} }  ~~ \sum\nolimits_{n \in {\cal N}} {r}[n]  \\
~{\rm{s.t.}}&~ ~\eqref{p1:asso1} ~{\text{and}}~\eqref{p1:asso2}. \nonumber
\end{align}
\end{subequations}
Though problem (P3) is an integer problem, it can be optimally solved by exploiting its special structure. In particular, for any UAV $k$ at slot $n$, it follows that the optimal solution of $\{\alpha_{m,k}[n]\}$ is obtained by ${\alpha^ * _{{m^ * },k}}[n] = 1$ and ${\alpha^ * _{m,k}}[n] = 0,\forall m \ne {m^ * }$, where ${m^ * } = \arg \mathop {\max }\limits_{\{ m\} } {r_{m,k}}[n]$. Notice that if there are multiple GBSs achieving the same maximum data rate, then we can choose any arbitrary one as $m^*$ without loss of optimality.

\subsection{Transmit Beamforming Optimization}
Next, we optimize the transmit beamforming $\{{\bf{W}}_{m,k}[n]\}$ and $\{{\bf{R}}_m[n]\}$ under given UAV trajectory $\{{\bf q}_k[n]\}$ and UAV association $\{{\alpha _{m,k}}[n]\}$, for which the optimization problem is
\begin{subequations}
\begin{align}
\left( {\rm{P4}} \right): &\mathop {\max }\limits_{\{ {\bf{W}}_{m,k}[n],{{\bf{R}}_m}[n]\} }  ~~  \sum\nolimits_{n \in {\cal N}} { r}[n]  \\
~{\rm{s.t.}}&~\eqref{P2:con:sen},~\eqref{P2:con:pow}, ~{\text{and}}~\eqref{P2:con:rank1}. \nonumber
\end{align}
\end{subequations}
Note that problem (P4) is non-convex due to the non-concavity of ${r}[n]$. To address this issue, we adopt the SCA technique to approximate the non-convex objective function as a series of concave ones. Consider each SCA iteration $o\ge 1$, in which the local point is denoted by $\{{\bf{W}}_{l,i}^{(o)}[n]\}$ and $\{{\bf{R}}_l^{(o)}[n]\}$. In this case, we approximate $r_{m,k}[n]$ in the objective function as its lower bound, i.e.,
\begin{align} \label{rate:SCA}
&{r_{m,k}}[n] \ge {\log _2}( \sum\nolimits_{l \in {\cal M}} \sum\nolimits_{i \in {\cal K}} {\rm{tr}}\left( {\bf{H}}_{m,k}[n]{\bf{W}}_{l,i}[n] \right) \nonumber\\
&+ \sum\nolimits_{l \in {\cal M}} {\rm{tr}}({\bf{H}}_{m,k}[n]{\bf{R}}_l[n])  + \sigma ^2 )\nonumber \\
&- a_{m,k}^{(o)}[n] - \sum\limits_{(l,i) \ne (m,k)} {\rm{tr}}({\bf{B}}_{m,k}^{(o)}[n]({\bf{W}}_{l,i}[n] - {\bf{W}}_{l,i}^{(o)}[n])) \nonumber \\
&- \sum\nolimits_{l\in {\cal M}} {\rm{tr}}({\bf{B}}_{m,k}^{(o)}[n]({\bf{R}}_l[n] - {\bf{R}}_l^{(o)}[n])) \triangleq {\bar r}_{m,k}^{(o)}[n],
\end{align}
where ${\bf{B}}_{m,k}^{(o)}[n]$ is defined in \eqref{deri:SCA} and
\begin{figure*}
\begin{align} \label{deri:SCA}
{\bf{B}}_{m,k}^{(o)}[n] = \frac{{\log }_2e{\bf{H}}_{m,k}[n]}{{\sum\nolimits_{(l,i) \ne (m,k)}  {{\rm{tr}}\left( {\bf{H}}_{m,k}[n]{\bf{W}}_{l,i}^{(o)}[n] \right)}   + \sum\nolimits_{l \in {\cal M}} {\rm{tr}}\left( {\bf{H}}_{m,k}[n]{\bf{R}}_l^{(o)}[n] \right)  + \sigma ^2}}
\end{align}
\vspace{-5ex}
\end{figure*}
\begin{align}\label{deri:SCAa}
&a_{m,k}^{(o)}[n]= {\log _2}(\sum\nolimits_{(l,i) \ne (m,k)}  {\rm{tr}}({\bf{H}}_{m,k}[n]{\bf{W}}_{l,i}^{(o)}[n]) \nonumber \\
& + \sum\nolimits_{l \in {\cal M}} {\rm{tr}}({\bf{H}}_{m,k}[n]{\bf{R}}_l^{(o)}[n]) + \sigma ^2).
\end{align}
By denoting ${\bar r}^{(o)}[n] = \sum\nolimits_{m \in {\cal M}} \sum\nolimits_{k \in {\cal K}} {\bar r_{m,k}^{(o)}} $ and replacing the objective function in (P4) as $\sum\nolimits_{n \in {\cal N}} {\bar r}^{(o)}[n]$, problem (P4) is approximated as (P5.$o$) in the $o$-th SCA iteration.

However, the rank-one constraints \eqref{P2:con:rank1} in problem (P5.$o$) make it still non-convex. To address this issue, we omit the rank-one constraints and obtain the SDR form of (P5.$o$) as (SDR5.$o$), which is a convex problem and can be optimally solved by standard convex optimization solvers such as CVX \cite{grant2014cvx}. Let $ \{ {\bf{W}}_{m,k}^{ * }[n] \}$ and $\{ {{\bf{R}}_m^ * }[n] \} $ denote the optimal solution to (SDR5.$o$). Notice that the obtained solutions of $\{ {{\bf{W}}_{m,k}^ * }[n] \}$ are generally with high ranks and thus may not satisfy the rank-one constraints in (P5.$o$). We have the following proposition to reconstruct equivalent rank-one solutions to problem (P5.$o$).

\begin{prop}
If the optimal solutions of $ \{ {\bf{W}}_{l,i}^ {*}[n] \} $ to problem (SDR5.$o$) are not rank-one, then we reconstruct the equivalent solutions to  problem (P5.$o$) as $\{ \{ {{\bf{\bar W}}}_{l,i}[n]\} ,\{ {\bf{\bar R}}_l[n]\} \}$ as
\begin{subequations}
\begin{align}
&{\bf{\bar w}}_{l,i}[n] \!= \!( {{\bf{h}}_{m,k}^H[n]{\bf{W}}_{l,i}^{*}[n]{{\bf{h}}_{m,k}[n]}} )^{ - \! \frac{1}{2}}{\bf{W}}_{l,i}^{*}[n]{{\bf{h}}_{m,k}[n]}, \label{app:1:1}  \\
&{{{\bf{\bar W}}}_{l,i}[n]} = {{\bf{\bar w}}_{l,i}[n]}{\bf{\bar w}}_{l,i}^H[n], \label{app:1:2} \\
&{\bf{\bar R}}_l[n] = \sum\nolimits_{i \in {\cal K}} {{\bf{W}}_{l,i}^{*}}[n]  + {\bf{R}}_l^{*}[n] - \sum\nolimits_{i \in {\cal K}} {\bf{\bar W}}_{l,i}[n], \label{app:1:3}
\end{align}
\end{subequations}
which satisfy the rank-one constraints and are feasible for problem (P5.$o$). The equivalent solutions in (19) achieve the same value for (P5.$o$) as the optimal value achieved by $ \{ {\bf{W}}_{m,k}^{ * }[n] \}$ and $\{ {{\bf{R}}_m^ * }[n] \} $ for problem (P5.o). Therefore, $\{{\bf{\bar w}}_{l,i}[n] \}$ and $\{{\bf{\bar R}}_l[n]\}$ are optimal for (P5.$o$) and the SDR (SDR5.$o$) of (P5.$o$) is tight.
\end{prop}

\begin{proof}
This proposition follows directly from \cite[Proposition 1]{Cheng2024Optimal}, for which the proof is omitted here.
\end{proof}

Therefore, in each SCA iteration $o$ we obtain the optimal solution to (P5.$o$), which can be shown to lead to a monotonically non-decreasing objective value of (P4). Hence, the SCA-based solution to problem (P4) is ensured to be converge.
\subsection{UAV Trajectory Optimization}
Next, we optimize the UAV trajectory $\{{\bf q}_k[n]\}$ with given transmit beamforming $\{{\bf{W}}_{l,i}[n]\}$, $\{{{\bf{R}}_l}[n]\}$ and UAV-GBS association $\{{ \alpha}_{l,i}[n]\}$. The corresponding optimization problem is
\begin{subequations}
\begin{align}
\left( {\rm{P6}} \right): &\mathop {\max }\limits_{\{ {\bf q}_k[n]\} }  ~~ \sum\nolimits_{n \in {\cal N}} {{ r}[n]}   \label{P6:obj} \\
~{\rm{s.t.}}&~\eqref{p1:ini}, ~\eqref{p1:end}, ~\eqref{p1:speed}, ~{\text{and}}~\eqref{p1:crash}. \nonumber
\end{align}
\end{subequations}
Note that in (P6), \eqref{p1:crash} are non-concave constraints and the optimization variables $\{{\bf q}_k[n]\}$ appear in the steering vector ${\bf g}_{m,k}[n]$, both of which are difficult to be tackled. First, we adopt SCA to approximate \eqref{p1:crash}. In each iteration $o$, the first-order Tyler expansion is applied to approximate the left-hand-side of \eqref{p1:crash} with respect to ${\bf{q}}_k^{(o)}[n]$ and ${\bf{q}}_i^{(o)}[n]$, based on which the approximate version of constraint \eqref{p1:crash} is expressed as
\begin{align}
&  2( {\bf{q}}_k^{(o)}[n] - {\bf{q}}_i^{(o)}[n] )^T \left( {\bf{q}}_k[n] - {\bf{q}}_i[n] \right) \nonumber \\
&- || {\bf{q}}_k^{(o)}[n] - {\bf{q}}_i^{(o)}[n] || ^2\ge D_{\min }^2 - (H_k-H_i)^2. \label{SCA:crash}
\end{align}

Next, we reformulate the objective function $r[n]$. To facilitate the reformulation, we denote the entries in the $p$-th and $q$-th column of ${\bf{W}}_{l,i}[n]$ and ${\bf{R}}_{l}[n]$ as $[ {\bf{W}}_{l,i}[n] ]_{p,q}$ and $\left[ {\bf{R}}_{l}[n] \right]_{p,q}$, respectively. Similarly, their absolute values are denoted by $| {\left[ {\bf{W}}_{l,i}[n] \right]}_{p,q} |$ and $| {\left[ {\bf{R}}_{l}[n] \right]}_{p,q} |$. Moreover, the phases of these entries are denoted by $\theta _{p,q}^{{\bf{W}}_{l,i}[n]}$ and $\theta _{p,q}^{{\bf{R}}_{l}[n]}$, respectively. Hence, we equivalently re-express $r_{m,k}[n]$ as
\begin{align}
&{r_{m,k}}[n] ={\log _2}(\sum\nolimits_{l \in {\cal M}} \sum\nolimits_{i \in {\cal K}} {\eta \left( {\bf{W}}_{l,i}[n],{\bf{q}}_k[n] \right)} \nonumber \\
& + \sum\nolimits_{l \in {\cal M}} \mu ({\bf{R}}_l[n],{\bf{q}}_k[n])  + \frac{{{\sigma ^2}}}{\kappa }({\left\| {{{\bf{q}}_k}[n] - {{\bf{u}}_m}} \right\|^2} + {H_k^2})) \nonumber \\
&  - {\log _2}(\sum\nolimits_{(l,i) \ne (m,k)} \eta ( {\bf{W}}_{l,i}[n],{\bf{q}}_k[n] )   \nonumber \\
& + \sum\limits_{l \in {\cal M}} \mu ({\bf{R}}_l[n],{\bf{q}}_k[n]) + \frac{{{\sigma ^2}}}{\kappa }({\left\| {{{\bf{q}}_k}[n] - {{\bf{u}}_m}} \right\|^2} + {H_k^2})), \label{traj:rate}
\end{align}
where
\begin{align}
&\eta ({\bf{W}}_{l,i}[n],{\bf{q}}_k[n])\!\! = \!\! \sum\limits_{r = 1}^{{N_a}}\! {[ {\bf{W}}_{l,i}[n] ]}_{r,r}\! \! +\! \! 2\sum\limits_{p = 1}^{N_a} \!  \sum\limits_{q = p+1}^{{N_a}} \! | {\left[ {\bf{W}}_{l,i}[n] \right]}_{p,q} | \nonumber \\
&\!\times \!\cos ( \theta _{p,q}^{{\bf{W}}_{l,i}[n]}+ 2\pi \frac{d}{\lambda }(q - p)\frac{H_k}{( {{\left\| {{{\bf{q}}_k}[n] \! -\! {{\bf{u}}_m}} \right\|}^2 \!+\! {H_k^2}})^{\frac{1}{2}} } ) ,
\end{align}
and
\begin{align}
&\mu ({\bf{R}}_l[n],{\bf{q}}_k[n])= \sum\limits_{r = 1}^{{N_a}} {\left[ {\bf{R}}_l [n] \right]}_{r,r}  + 2\sum\limits_{p = 1}^{N_a} \sum\limits_{q = p+1}^{N_a} | {\left[ {\bf{R}}_l[n] \right]}_{p,q} | \nonumber \\
&\!\times \!\cos ( \theta _{p,q}^{{\bf{R}}_{l}[n]}\! +\! 2\pi \frac{d}{\lambda }(q - p)\frac{H_k}{( {{\left\| {{{\bf{q}}_k}[n] - {{\bf{u}}_m}} \right\|}^2 \!+\! {H_k^2}})^{\frac{1}{2}} } ).
\end{align}
In the following, we adopt first-order Tyler expansion to approximate the re-expressed \eqref{traj:rate} in iteration $o$ as \eqref{QoS:linear}
\begin{align}
{r_{m,k}}[n] \! \approx \! {{\tilde r}^{(o)}_{m,k}}[n] \!\! = \! \!  c_{m,k}^{(o)}[n] \!+ \! {\bf{d}}_{m,k}^{(o)T}[n]({\bf{q}}_k[n] \!-\! {\bf{q}}_k^{(o)}[n]), \label{QoS:linear}
\end{align}
where $c_{m,k}^{(o)}[n] $ and ${\bf{d}}_{m,k}^{(o)}[n]$ are defined in \eqref{SCA:cmk} and \eqref{SCA:dmk} with $\nu ({\bf{W}}_{l,i}[n],{\bf{q}}_k^{(o)}[n])$, $\upsilon ({\bf{R}}_l[n],{\bf{q}}_k^{(o)}[n])$, ${g_{m,k}}[n]$, and ${h_{m,k}}[n]$ given in \eqref{SCA:nu}, \eqref{SCA:upsilon}, \eqref{SCA:gmk}, and \eqref{SCA:hmk}, respectively.
\begin{align}\label{SCA:cmk}
&c_{m,k}^{(o)}[n]\! \!=\! \!{\log _2}(\!\sum\limits_{l \in {\cal M}}\! \sum\limits_{i \in {\cal K}} \!{\eta ( {\bf{W}}_{l,i}[n] \!, \!{\bf{q}}_k^{(o)}[n] )} \! \!+ \! \! \! \sum\limits_{l \in {\cal M}} \! \! \mu ({\bf{R}}_l[n] \! , \!{\bf{q}}_k^{(o)}\![n])  \nonumber \\
&+ \! \frac{{{\sigma ^2}}}{\kappa}({ || {{{\bf{q}}_k^{(o)}}[n] \! \! - \! \! {{\bf{u}}_m}} || ^2} \! \! + \! \!  {H_k^2})) \!-\! {\log _2}(\sum\limits_{\scriptstyle(l,i) \ne \hfill\atop
\scriptstyle(m,k)\hfill} \! \! {\eta ({{\bf{W}}_{l,i}}[n],{\bf{q}}_k^{(o)}[n])}  \nonumber \\
&+ \!\sum\limits_{l \in {\cal M}} \!\mu ({\bf{R}}_l[n],{\bf{q}}_k^{(o)}[n])\! + \! \frac{{{\sigma ^2}}}{\kappa }({ || {{{\bf{q}}_k^{(o)}}[n] - {{\bf{u}}_m}}|| ^2} \!+ \!{H_k^2})),
\end{align}
\begin{align}\label{SCA:dmk}
&{\bf{d}}_{m,k}^{(o)}[n] = \frac{{ {{\log }_2}e}}{g_{m,k}[n]}(\sum\nolimits_{l \in {\cal M}} {\sum\nolimits_{i \in {\cal K}} {\nu ({\bf{W}}_{l,i}[n],{\bf{q}}_k^{(o)}[n])} } \nonumber \\
&+ \sum\nolimits_{l \in {\cal M}} {\upsilon ({{\bf{R}}_l}[n],{\bf{q}}_k^{(o)}[n])} )({{\bf{q}}_k^{(o)}}[n] - {{\bf{u}}_m}) \nonumber \\
& - \frac{{ {{\log }_2}e}}{h_{m,k}[n]}(\sum\nolimits_{(l,i) \ne (m,k)} {\nu ({{\bf{W}}_{l,i}}[n],{\bf{q}}_k^{(o)}[n])} \nonumber \\
&+ \sum\nolimits_{l \in {\cal M}} {{\upsilon ({{\bf{R}}_l}[n],{\bf{q}}_k^{(o)}[n])}} )({{\bf{q}}_k^{(o)}}[n] - {{\bf{u}}_m}),
\end{align}
\begin{align}\label{SCA:nu}
&\nu ({\bf{W}}_{l,i}[n],{\bf{q}}_k^{(o)}[n])= \sum\limits_{p = 1}^{{N_a}} \sum\limits_{q = p+1}^{{N_a}} 4\pi | {\left[ {\bf{W}}_{l,i}[n] \right]}_{p,q} | \nonumber \\
&\times \sin ( \theta _{p,q}^{{\bf{W}}_{l,i}[n]}+2\pi \frac{d}{\lambda }(q - p)\frac{H_k}{( {{ || {\bf{q}}_k^{(o)}[n] - {{\bf{u}}_m}|| }^2 + {H_k^2}} )^{\frac{1}{2}}}) \nonumber \\
&\times \frac{d H_k(q - p) }{\lambda ( {{ || {{{\bf{q}}_k^{(o)}}[n] - {{\bf{u}}_m}} || }^2 + {H_k^2}} )^{\frac{3}{2}}} ,
\end{align}
\begin{align}\label{SCA:upsilon}
&\upsilon ({\bf{R}}_l[n],{\bf{q}}_k^{(o)}[n])= \sum\limits_{p = 1}^{{N_a}} \sum\limits_{q = p+1}^{{N_a}} 4\pi | {{{\left[ {{\bf{R}}_l}[n] \right]}_{p,q}}} | \nonumber \\
&\times \sin ( \theta _{p,q}^{{{\bf{R}}_l}[n]}+ 2\pi \frac{d}{\lambda }(q - p)\frac{H_k}{( {{ || {{{\bf{q}}_k^{(o)}}[n] - {{\bf{u}}_m}}|| }^2 + {H_k^2}} )^{\frac{1}{2}}})\nonumber \\
& \times \frac{d H_k(q - p)}{\lambda ( {{ || {{{\bf{q}}_k^{(o)}}[n] - {{\bf{u}}_m}}|| }^2 + {H_k^2}} )^{\frac{3}{2}}},
\end{align}
\begin{align}\label{SCA:gmk}
&{g_{m,k}}[n]= \sum\nolimits_{l \in {\cal M}} \sum\nolimits_{i \in {\cal K}} \eta ( {\bf{W}}_{l,i}[n],{\bf{q}}_k^{(o)}[n] )\!+ \! \sum\nolimits_{l \in {\cal M}}\!\nonumber \\
& \mu ({\bf{R}}_l[n],{\bf{q}}_k^{(o)}[n]) \!+\! \frac{{{\sigma ^2}}}{\kappa }({ || {{{\bf{q}}_k^{(o)}}[n] - {{\bf{u}}_m}} || ^2} \!+\! {H_k^2}),
\end{align}
\begin{align}\label{SCA:hmk}
&{h_{m,k}}[n] = \sum\nolimits_{(l,i) \ne (m,k)} \eta ( {\bf{W}}_{l,i}[n],{\bf{q}}_k^{(o)}[n] ) \!+\! \sum\nolimits_{l \in {\cal M}} \nonumber \\
&\mu ({{\bf{R}}_l[n]},{\bf{q}}_k^{(o)}[n]) + \frac{{{\sigma ^2}}}{\kappa }({\left\| {{{\bf{q}}_k}[n] - {{\bf{u}}_m}} \right\|^2} + {H_k^2}).
\end{align}

To ensure the accuracy of our  approximation \eqref{traj:rate} and \eqref{QoS:linear}, we consider a series of trust region constraints in each iteration $o$:
\begin{align}
|| {{\bf{q}}_k^{( o)}[n] - {\bf{q}}_k^{( o - 1)}[n]}||  \le {\omega ^{( o)}}, \forall k \in {\cal K}, n \in {\cal N}, \label{trust:region}
\end{align}
where ${\omega ^{( o)}}$ denotes the radius of the trust region. Let ${{\tilde r}^{(o)}}[n] = \sum\nolimits_{l \in {\cal M}} {\sum\nolimits_{i \in {\cal K}} {{\alpha _{m,k}}[n]\tilde r_{m,k}^{(o)}[n]} } $ , by replacing $r[n]$ as ${{\tilde r}^{(o)}}[n]$, we obtain the approximated version of UAV trajectory optimization (P6) as (P7.$o$)

\begin{subequations}
\begin{align}
\left( {\rm{P7}}.o \right): &\mathop {\max }\limits_{\{ {\bf q}_k[n]\} }  ~~ \sum\nolimits_{n \in {\cal N}} {{\tilde r}^{(o)}}[n]   \label{P7:obj} \\
~{\rm{s.t.}}&~\eqref{p1:ini}, ~\eqref{p1:end}, ~\eqref{p1:speed}, ~\eqref{SCA:crash},~{\text{and}}~\eqref{trust:region}.
\end{align}
\end{subequations}

In summary, we solve problem (P6) by iteratively solving a series of convex problems (P7.$o$)'s. The convergence of approximation \eqref{QoS:linear} can be ensured by choosing a sufficiently small ${\omega ^{( o)}}$. In practice, in each iteration $o$, we reduce the trust region radius by ${\omega ^{(o)}} = \frac{1}{2}{\omega ^{(o)}}$ if the objective value is non-decreasing. The reduction will be terminated when ${\omega ^{(o)}}$ is lower than a threshold.

\begin{figure*}[htbp]
\centering
\subfigure[$\Gamma$=-20 dBW, N=10, H=80 m.]{\includegraphics[width=4.45cm]{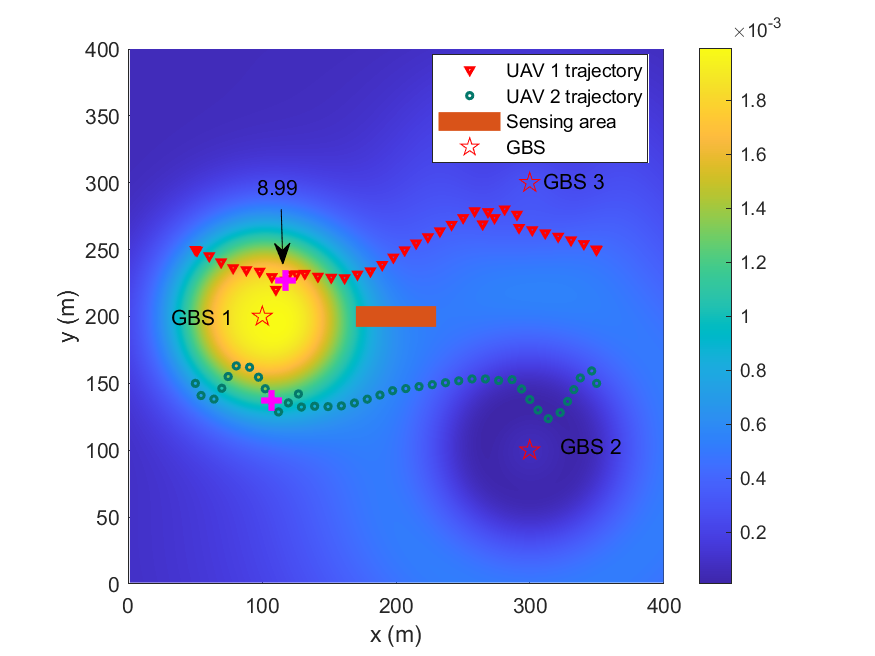}}
\subfigure[$\Gamma$=-20 dBW, N=10, H=100 m.]{\includegraphics[width=4.45cm]{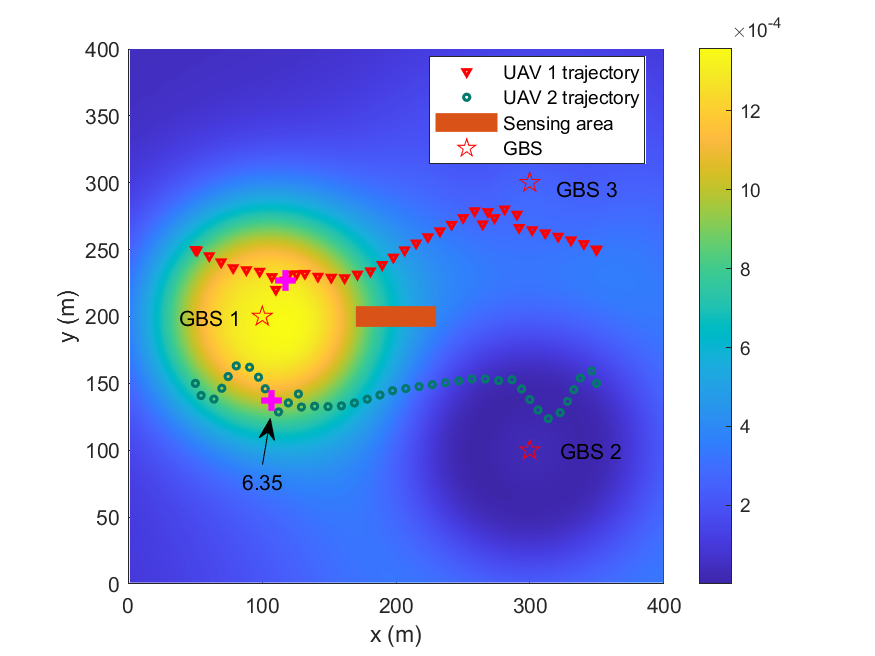}}
\subfigure[$\Gamma$=-10 dBW, N=10, H=80 m.]{\includegraphics[width=4.45cm]{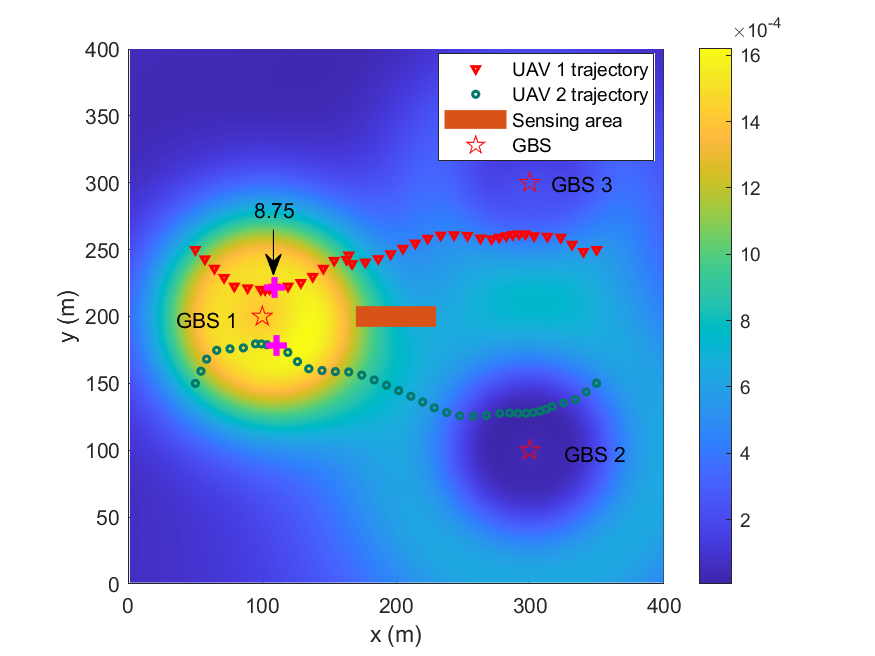}}
\subfigure[$\Gamma$=-10 dBW, N=10, H=100 m.]{\includegraphics[width=4.45cm]{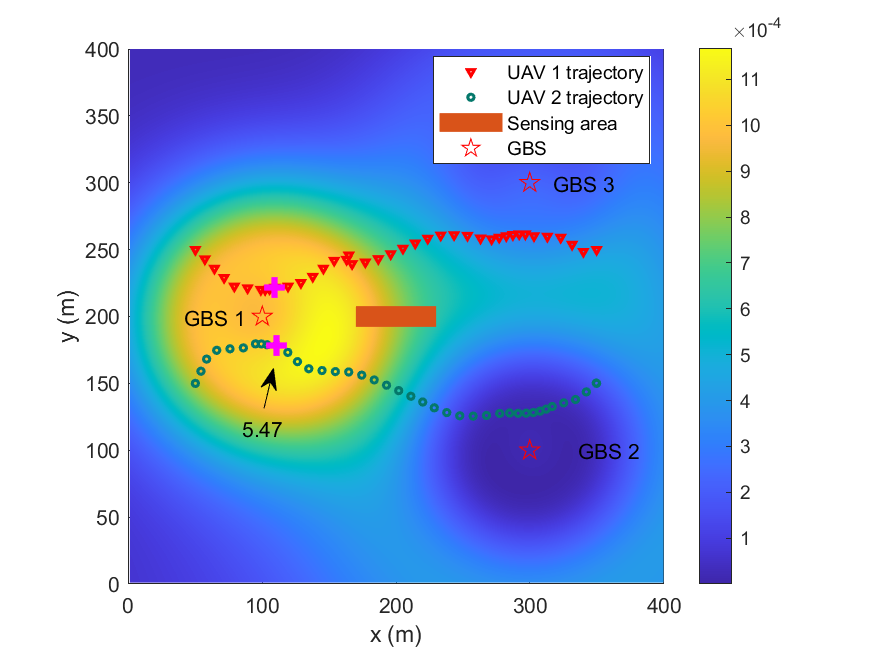}}
\caption{The achieved beampattern gain via proposed design in the corresponding altitudes of UAVs, the chosen time slot is $N=10$, the carmine '+' denote the UAVs location of this time slot, and the numbers associated with each UAV correspond to its communication rate (in bps/Hz).}
\label{fig:bpg}
\vspace{-4ex}
\end{figure*}
\section{Numerical Results}

In this section, we provide numerical results to illustrate the performance of our proposed designs for supporting networked ISAC for UAVs in LAE. In the simulation, we consider a $400~{\text{m}} \times 400~ {\text{m}}$ area with $M=3$ GBSs, $K=2$ UAVs, and $Q=20$ sample locations. The antenna spacing is set as $d = \frac{\lambda }{2}$, and the number of antenna at GBSs is $N_a=4$. We set the initial locations of UAVs as ${\bf{q}}_1^{\rm{I}} = [50{\text{m}},250{\text{m}}]$ and  ${\bf{q}}_2^{\rm{I}} = [50 {\text{m}},150{\text{m}}]$, the final locations as ${\bf{q}}_1^{\rm{F}} = [350{\text{m}},250{\text{m}}]$ and ${\bf{q}}_2^{\rm{F}} = [350{\text{m}},150{\text{m}}]$, the maximum speed of UAVs as $V_{\max}=10~ {\text{m/s}}$, and the flight altitude as ${H_1} = 80~{\text{m}}$ and ${H_2} = 100~{\text{m}}$, the total number of slots as $N=40$, and the maximum transmit power at GBSs as $P_{\max}=3 ~{\text{W}}$. Furthermore, the channel power gain at reference distance $d_0=1 {\text{m}}$ is $\kappa=-45~ {\text{dB}}$, and the noise power at UAVs is $\sigma^2=-100~{\text{dBW}}$. For comparison, we consider two benchmark designs as follows.

{\bf{Transmit beamforming only with straight flight}}: The authorized UAVs take off from the initial locations ${\bf{q}}_k^{\rm{I}}$ and landing at the final locations ${\bf{q}}_k^{\rm{F}}$ with straight flight by using the constant speed ${V_k} = {\textstyle{1 \over N}}\left\| {{\bf{q}}_k^{\rm{I}} - {\bf{q}}_k^{\rm{F}}} \right\|$. Accordingly, the transmit beamforming $\{{{\bf{w}}_{l,i}}[n]\}$ and $\{{{\bf{R}}_{l}}[n]\}$, and UAV association $\{\alpha_{l,i}[n]\}$ are optimized by solving (P2) with predetermined straight trajectories.

{\bf{Joint power allocation and trajectory design}}: The GBSs employ the isotropic transmission with ${\tilde{\bf{W}}_{l,i}}[n] = {\textstyle{{p_{l,i}^c[n]} \over {{N_a}}}}{\bf{I}}$ and ${\tilde{\bf{R}}_l}[n] = {\textstyle{{p_l^s[n]} \over {{N_a}}}}{\bf{I}}$, where ${p_{l,i}^c}[n]$ and ${p_l^s}[n]$ denote the transmit power values of information signals and dedicated sensing signals, respectively, which are optimization variables. Similarly, the power constraint at each GBS $l$ becomes $\sum\nolimits_{i \in {\cal K}} {p_{l,i}^c[n]}  + p_l^s[n] \le {P_{\max }}$. By substituting $\{{\tilde{\bf{W}}_{l,i}}[n]\}$ and $\{{\tilde{\bf{R}}_l}[n]\}$ into (P2) by omitting the rank constraints, we use the AO-based algorithm similarly to that in Section \ref{Sec:solution} to solve it.

Fig. \ref{fig:bpg} shows the obtained UAV trajectories and illumination powers of our proposed design at time slot $N=10$. It is observed that as $\Gamma$ increases from $-20$ dBm to $-10$ dBm, the UAVs fly more closely to the GBSs to obtain decreased path loss. This is because when the sensing constraints become critical, the GBSs need to steer the transmit beams towards the sensing area, and thus the UAV may need to fly close to that area to exploit the strong signal power. Comparing Fig. \ref{fig:bpg}(a) with \ref{fig:bpg}(b), it is observed that UAV $1$ achieves a higher communication rate than UAV $2$. Thanks to the properly designed transmit beamforming and the flexibly designed trajectory as well as UAV-GBS association, UAV $1$ is associated with GBS $1$ and hovers over slots $8-12$ for enhanced communication rate, while UAV $2$ is associated with GBS $2$ to avoid the inter-cell interference. We have similar observations in Figs. \ref{fig:bpg}(c) and \ref{fig:bpg}(d).

Fig. \ref{fig:tradeoff} shows the achieved average sum rate of authorized UAVs versus the illumination power for sensing. It is observed that as $\Gamma$ increases, the UAVs' average sum rate declines for all schemes. This is because GBSs need to spend more transmit power to satisfy the sensing requirements. It is also observed that the proposed design outperforms other benchmarks when $\Gamma$ varies from $-25$dBW to $-8$dBW, which validates the benefits of our trajectory design and transmit beamforming optimization. Furthermore, the benchmark of joint power allocation and trajectory design is observed to perform significantly worse than other schemes, which becomes infeasible when $\Gamma > 13$ dBW. This is because this scheme lacks the ability to reshape the transmit beamformers to cater to increased sensing requirements due to the limited DoFs.
\begin{figure}[ht]
\centering
    \includegraphics[width=6cm]{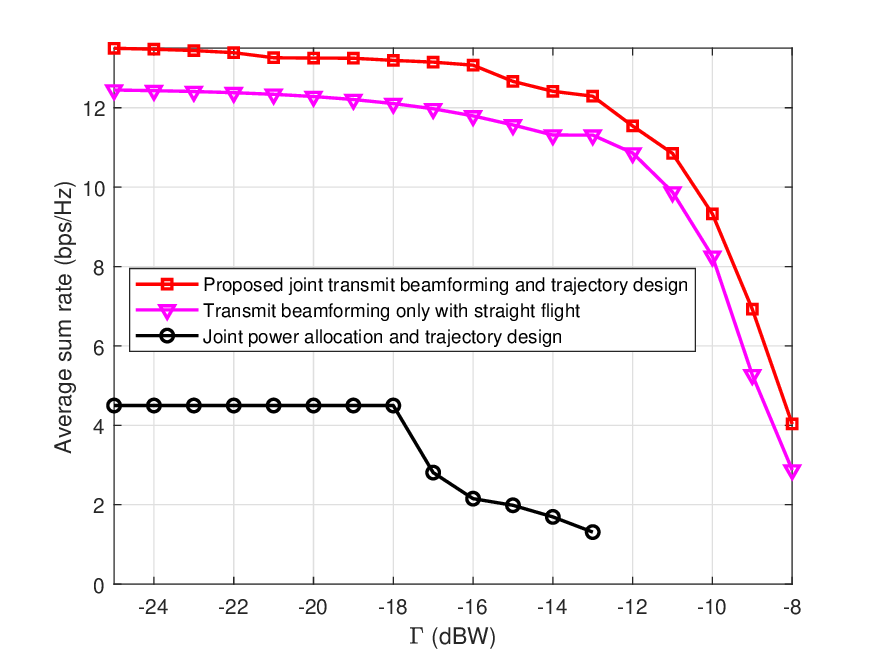}
\caption{The average sum rate of authorized UAVs versus the illumination power for sensing.}
\label{fig:tradeoff}
\end{figure}


\section{Conclusion}
This paper considered a networked ISAC system to support LAE, in which a set of networked GBSs cooperatively communicate with multiple authorized UAVs and concurrently perform cooperative sensing to sense the interested airspace for monitoring unauthorized objects. We proposed the joint cooperative transmit beamforming and trajectory design to maximize the average sum rate of authorized UAVs while satisfying the sensing requirements. We proposed an efficient algorithm to solve the highly non-convex problem via using the AO, SDR, and SCA techniques. Numerical results showed that our proposed joint transmit beamforming and trajectory design significantly outperforms other benchmark designs.

\bibliographystyle{IEEEtran}
\bibliography{IEEEabrv,myref}

\begin{thebibliography}{10}
\providecommand{\url}[1]{#1}
\csname url@samestyle\endcsname
\providecommand{\newblock}{\relax}
\providecommand{\bibinfo}[2]{#2}
\providecommand{\BIBentrySTDinterwordspacing}{\spaceskip=0pt\relax}
\providecommand{\BIBentryALTinterwordstretchfactor}{4}
\providecommand{\BIBentryALTinterwordspacing}{\spaceskip=\fontdimen2\font plus
\BIBentryALTinterwordstretchfactor\fontdimen3\font minus
  \fontdimen4\font\relax}
\providecommand{\BIBforeignlanguage}[2]{{%
\expandafter\ifx\csname l@#1\endcsname\relax
\typeout{** WARNING: IEEEtran.bst: No hyphenation pattern has been}%
\typeout{** loaded for the language `#1'. Using the pattern for}%
\typeout{** the default language instead.}%
\else
\language=\csname l@#1\endcsname
\fi
#2}}
\providecommand{\BIBdecl}{\relax}
\BIBdecl

\bibitem{Whitepaper4}
\BIBentryALTinterwordspacing
{China Telecom, Ericsson, Nokia, Huawei, ZTE, CICT Mobile, OPPO, Xiaomi, vivo,
  Lenovo, Qualcomm, Mediatek, UNISOC}, ``The low-altitude network by integrated
  sensing and communication,'' White Paper, Feb. 2024. [Online]. Available:
  \url{https://www.zte.com.cn/content/dam/zte-site/res-www-zte-com-cn/mediares/zte/%E6%97%A0%E7%BA%BF%E6%8E%A5%E5%85%A5/%E7%99%BD%E7%9A%AE%E4%B9%A6/Low_altitude_network_by_ISAC.pdf}
\BIBentrySTDinterwordspacing

\bibitem{Mu2023magazine}
J.~Mu, R.~Zhang, Y.~Cui, N.~Gao, and X.~Jing, ``{UAV} meets integrated sensing
  and communication: Challenges and future directions,'' \emph{IEEE Commun.
  Mag.}, vol.~61, no.~5, pp. 62--67, Jan. 2023.

\bibitem{zhang2021perceptive}
A.~Zhang, M.~L. Rahman, X.~Huang, Y.~J. Guo, S.~Chen, and R.~W. Heath,
  ``Perceptive mobile networks: Cellular networks with radio vision via joint
  communication and radar sensing,'' \emph{IEEE Veh. Technol. Mag.}, vol.~16,
  no.~2, pp. 20--30, Jun. 2021.

\bibitem{Fan2023Sensing}
F.~Dong, F.~Liu, Y.~Cui, W.~Wang, K.~Han, and Z.~Wang, ``Sensing as a service
  in 6{G} perceptive networks: A unified framework for {ISAC} resource
  allocation,'' \emph{IEEE Trans. Wireless Commun.}, vol.~22, no.~5, pp.
  3522--3536, May 2023.

\bibitem{Fei2023magazine}
Z.~Fei, X.~Wang, N.~Wu, J.~Huang, and J.~A. Zhang, ``Air-ground integrated
  sensing and communications: Opportunities and challenges,'' \emph{IEEE
  Commun. Mag.}, vol.~61, no.~5, pp. 55--61, Feb. 2023.

\bibitem{Meng2023magazine}
K.~Meng, Q.~Wu, J.~Xu, W.~Chen, Z.~Feng, R.~Schober, and A.~L. Swindlehurst,
  ``{UAV}-enabled integrated sensing and communication: Opportunities and
  challenges,'' \emph{IEEE Wireless Commun.}, vol.~31, no.~2, pp. 97--104, Apr.
  2023.

\bibitem{Meng2023Throughput}
K.~Meng, Q.~Wu, S.~Ma, W.~Chen, K.~Wang, and J.~Li, ``Throughput maximization
  for {UAV}-enabled integrated periodic sensing and communication,'' \emph{IEEE
  Trans. Wireless Commun.}, vol.~22, no.~1, pp. 671--687, Aug. 2022.

\bibitem{Lyu2022Joint}
Z.~Lyu, G.~Zhu, and J.~Xu, ``Joint maneuver and beamforming design for
  {UAV}-enabled integrated sensing and communication,'' \emph{IEEE Trans.
  Wireless Commun.}, vol.~22, no.~4, pp. 2424--2440, Oct. 2022.

\bibitem{wang2020constrained}
X.~Wang, Z.~Fei, J.~A. Zhang, J.~Huang, and J.~Yuan, ``Constrained utility
  maximization in dual-functional radar-communication multi-{UAV} networks,''
  \emph{IEEE Trans. Commun.}, vol.~69, no.~4, pp. 2660--2672, Apr. 2021.

\bibitem{Wu2023uav}
J.~Wu, W.~Yuan, and L.~Bai, ``On the interplay between sensing and
  communications for {UAV} trajectory design,'' \emph{IEEE Internet Things J.},
  vol.~10, no.~23, pp. 20\,383--20\,395, Jun. 2023.

\bibitem{grant2014cvx}
\BIBentryALTinterwordspacing
M.~Grant and S.~Boyd, ``{CVX}: Matlab software for disciplined convex
  programming, version 2.1,'' 2014. [Online]. Available:
  \url{http://cvxr.com/cvx/}
\BIBentrySTDinterwordspacing

\bibitem{Cheng2024Optimal}
G.~Cheng, Y.~Fang, J.~Xu, and D.~W.~K. Ng, ``Optimal coordinated transmit
  beamforming for networked integrated sensing and communications,'' \emph{IEEE
  Trans. Wireless Commun.}, pp. 1--1, early access, Jan. 2024.

\end{thebibliography}

\end{document}